\documentclass[12pt, a4paper]{article}      % onecolumn (second format)
\usepackage[dvips]{graphicx}
\usepackage{wrapfig}
\usepackage{amssymb}
\usepackage{amsmath}
\usepackage[a4paper,hmargin={3cm,.9in},height=10in]{geometry}
\usepackage{mathrsfs}
\usepackage{eurosym}
\usepackage{dsfont}
\begin{document}

\title{Quantum
information approach to normal representation of extensive games%\thanks{Grants or other notes
%about the article that should go on the front page should be
%placed here. General acknowledgments should be placed at the end of the article.}
}

%\titlerunning{Short form of title}        % if too long for running head

\author{Piotr Fr\c{a}ckiewicz}
\newtheorem{lemma}{Lemma}
\newtheorem{definition}{Definition}[section]
\newtheorem{theorem}{Theorem}
\newtheorem{proposition}{Proposition}[section]
\newtheorem{example}[proposition]{Example}
\newenvironment{proof}{\noindent\textit{Proof.}}
{\nolinebreak[4]\hfill$\blacksquare$\\\par}
%\authorrunning{Short form of author list} % if too long for running head

\author{\textsc{Piotr Fr\c{a}ckiewicz} \\Institute of Mathematics of the Polish Academy of Sciences\\
00-956 Warsaw, Poland}

\maketitle

\begin{abstract}
We modify the concept of quantum strategic game to make it useful
for extensive form games. We prove that our modification allows to
consider the normal representation of any finite extensive game
using the fundamental concepts of quantum information. The
Selten's Horse game and the general form of two-stage extensive
game with perfect information are studied to illustrate a
potential application of our idea. In both examples we use
Eisert-Wilkens-Lewenstein approach as well as Marinatto-Weber
approach to quantization of games.\\
% \PACS{91A18 \and PACS code2 \and more}

\noindent {\bf Keywords} \quad Extensive game $\cdot$ Normal
representation $\cdot$ Quantum game $\cdot$ Nash equilibrium

\noindent {\bf Mathematics Subject Classification (2000)} \quad
81P68~$\cdot$~91A18~$\cdot$~91A80 \end{abstract}
\section{Introduction}
Over the period of twelve years of research on quantum games
\cite{survey}, the idea of quantum strategic $2 \times 2$ game
\cite{eisert2} has been well established. From mathematical point
of view, the quantum information approach to a $2\times2$ game is
described by the four-tuple:
\begin{equation}
\label{fourtuple}\left(\mathscr{H}, |\psi_{\mathrm{in}}\rangle,
\{\mathcal{U}_{i}\}, \{E_{i}\}\right).
\end{equation}
The Hilbert space $\mathscr{H} = \mathds{C}^2 \otimes
\mathds{C}^2$ is the place of the game and the sets of unitary
operators $\mathcal{U}_{1}, \mathcal{U}_{2} \subseteq
\mathsf{SU}(2)$ play the role of strategy sets for the first and
the second player, respectively. Given a unit vector
$|\psi_{\mathrm{in}}\rangle \in \mathscr{H}$ ({\it the initial
state}), the players each choose a unitary operator $U_{i} \in
\mathcal{U}_{i}$ changing the vector $|\psi_{\mathrm{in}}\rangle
\in \mathscr{H}$ into the vector
$|\psi_{\mathrm{fin}}\rangle\mathrel{\mathop:}= \left(U_{1}\otimes
U_{2}\right)|\psi_{\mathrm{in}}\rangle$ ({\it the final state}).
The last components are the functionals $E_{i}\colon \mathscr{H}
\to \mathds{R}$ for $i=1,2$. They imitate payoff functions for the
players assigning a real number to the final state
$|\psi_{\mathrm{fin}}\rangle$. It turns out that the four tuple
(\ref{fourtuple}) generalizes playing a classical $2 \times 2$
game. The two well-known ways based on the framework
(\ref{fourtuple}): the Eisert-Wilkens-Lewenstein (EWL) scheme
\cite{eisert} and the Marinatto-Weber (MW) scheme \cite{marinatto}
show that a $2\times2$ game can be successfully written in the
form (\ref{fourtuple}): it is possible to set each of the four
components, so that a $2\times2$ game and the corresponding game
defined by the four-tuple (\ref{fourtuple}) are the same with
respect to a game-theoretic analysis. Another key feature is that
the protocol (\ref{fourtuple}) allows to achieve results
unavailable in the game played classically (see, for example,
\cite{flitney0} and \cite{fracor1}). In our paper we are going to
deal with extensive games in quantum domain. In spite of quite
many researches connected with this issue (for instance,
concerning quantum Stackelberg duopoly \cite{iqbalbackwards} and
\cite{khan}), there is no generally accepted framework for playing
quantum extensive games by now. Interestingly, we have shown in
\cite{fracor2} that (\ref{fourtuple}) may indeed be useful for
extensive games with imperfect information. In this paper, we
extend our previous idea. We prove that the slight modification of
$\left(\mathscr{H}, |\psi_{\mathrm{in}}\rangle,
\{\mathcal{U}_{i}\}, \{E_{i}\}\right)$ allows to obtain normal
representation of extensive games in the quantum domain.
\section{Preliminaries to game theory}
Definitions in the preliminaries are based on \cite{osborne}. This
section starts with a~definition of a~finite extensive game
(without chance moves).
\begin{definition} Let the following components be given:
\begin{itemize}
\item A finite set $N = \{1,2,\dots,n\}$ of players. \item A
finite set $H$ of finite sequences that satisfies the following
two properties:
\begin{itemize}
\item the empty sequence $\emptyset$ is a member of $H$; \item if
$(a_k)_{k = 1,2,\dots, K} \in H$ and $L<K$ then $(a_k)_{k =
1,2,\dots, L} \in H$.
\end{itemize}
Each member of $H$ is a history and each component of a history is
an action taken by a player. A history $(a_{1}, a_{2},\dots,
a_{K}) \in H$ is terminal if there is no $a_{K+1}$ such that
$(a_{1}, a_{2},\dots, a_{K}, a_{K+1}) \in H$. The set of actions
available after the nonterminal history $h$ is denoted $A(h) = \{a
\colon (h,a) \in H\}$ and the set of terminal histories is denoted
$Z$. \item The player function $P \colon H \setminus  Z
\rightarrow N$ that points to a player who takes an~action after
the history $h$. \item For each player $i\in N$ a partition
$\mathcal{I}_{i}$ of $\{h \in H \setminus Z: P(h) = i\}$ with the
property that for each $I_{i} \in \mathcal{I}_{i}$ and for each
$h$, $h'$ $\in I_{i}$  an equality $A(h) = A(h')$ is fulfilled.
Every information set $I_{i}$ of the partition corresponds to the
state of player's knowledge. When the player makes move after
certain history $h$ belonging to $I_{i}$, she knows that the
course of events of the game takes the form of one of histories
being part of this information set. She does not know, however, if
it is the history $h$ or the other history from $I_{i}$. \item For
each player $i \in N$ a utility function $u_{i}\colon Z \to
\mathds{R}$ which assigns a number (payoff) to each of the
terminal histories.
\end{itemize}
A five-tuple $\left(N, H, P, \{\mathcal{I}_{i}\},
\{u_{i}\}\right)$ is called a finite extensive game.
\label{edefinition}
\end{definition}
Our deliberations focus on games with perfect recall (although
Def.~\ref{edefinition} defines extensive games with imperfect
recall as well) - this means games in which at each stage every
player remembers all the information about a course of the game
that she knew earlier (see \cite{osborne} and \cite{myerson} to
learn about formal description of this feature).

The notions: action and strategy mean the same in static games,
because players choose their actions once and simultaneously. In
the majority of extensive games a player can make her decision
about an action depending on all the actions taken previously by
herself and also by all the other players. In other words, players
can make some plans of actions at their disposal such that these
plans point out to a specific action depending on the course of
a~game. Such a plan is defined as a strategy in an extensive game.
\begin{definition}
A pure strategy $s_{i}$ of a player $i$ in a game $(N, H, P,
\{\mathcal{I}_{i}\}, \{u_{i}\})$ is a~function that assigns an
action in $A(I_{i})$ to each information set $I_{i} \in
\mathcal{I}_{i}$. \label{strategy}
\end{definition}
Like in the theory of strategic games, {\em a mixed strategy}
$t_{i}$ of a player $i$ in an extensive game is a probability
distribution over the set of player $i$'s pure strategies.
Therefore, pure strategies are of course special cases of mixed
strategies and from this place whenever we shall write {\em
strategy} without specifying that it is either pure or mixed, this
term will cover both cases. Let us define an {\em outcome $O(s)$}
of a pure strategy profile $s = (s_{1}, s_{2},\dots, s_{n})$ in an
extensive game without chance moves to be a terminal history that
results if each player $i\in N$ follows the plan of $s_{i}$. More
formally, $O(s)$ is the history $(a_{1}, a_{2},\dots, a_{K}) \in
Z$ such that for $0 \leq k < K$ we have $s_{P(a_{1}, a_{2},\dots,
a_{k})}(a_{1}, a_{2},\dots, a_{k}) = a_{k+1}$.
\begin{definition}
Let an extensive game $\mathrm{\Gamma} = \left(N, H, P,
\{\mathcal{I}_{i}\}, \{u_{i}\}\right)$ be given. The normal
representation of $\mathrm{\Gamma}$ is a~strategic game $\left( N,
\{S_{i}\}, \{u_{i}'\} \right)$ in which for each player $i \in N$:
\begin{itemize}
\item $S_{i}$ is the set of pure strategies of a player $i$ in
$\mathrm{\Gamma}$; \item $u_{i}'\colon \prod_{i \in N}S_{i} \to
\mathds{R}$ defined as $u_{i}'(s)\mathrel{\mathop:}=u_{i}(O(s))$
for every $s \in \prod_{i \in N}S_{i}$ and $i \in N$.
\end{itemize}
\end{definition}
One of the most important notions in game theory is a notion of an
equilibrium introduced by John Nash in \cite{nash}. A Nash
equilibrium is a profile of strategies where the strategy of each
player is optimal if the choice of its opponents is fixed. In
other words, in the equilibrium none of the players has any reason
to unilaterally deviate from an equilibrium strategy. A precise
formulation is as follows:
\begin{definition}
Let $(N, \{S_{i}\}, \{u_{i}\})$ be a game in strategic form. A
profile of strategies $(t^*_{1}, t^*_{2},\dots,t^*_{n})$ is a Nash
equilibrium if for each player $i \in N$ and for all $s_{i} \in
S_{i}$:
\begin{equation}
u_{i}(t^*_{i}, t^*_{-i}) \geq u_{i}(s_{i}, t^*_{-i}) ~~
\mbox{where} ~~  t^*_{-i} =
(t^*_{1},\dots,t^*_{i-1},t^*_{i+1},\dots,t^*_{n}).\label{nashequation}
\end{equation} \label{nashequilibrium}
\end{definition}
A Nash equilibrium in an extensive game with perfect recall is a
Nash equilibrium of its normal representation, hence
Def.~\ref{nashequilibrium} applies to strategic games as well as
to extensive ones.
\section{Preliminaries to quantum computing}
In this section we give a brief overview of the Dirac notation and
basic terms of quantum information. The preliminaries are based on
\cite{nielsen} and are sufficient to study the paper. Nonetheless,
we encourage the reader unfamiliar with techniques from theory of
quantum information to consult \cite{nielsen} and, for example,
\cite{kaye}.

First of all we adopt the convention that instead of denoting
vectors by boldface letters, e.g. $\mathbf{v}$, they are denoted
as {\it kets}: $|v\rangle$.

Let $\mathds{C}^{m+1}$ be a vector space with the fixed basis
$\{|v_{0}\rangle, |v_{1}\rangle,\dots,|v_{m}\rangle\}$ and let
\begin{equation}
|\phi\rangle = a_{0}|v_{0}\rangle + a_{1}|v_{1}\rangle + \dots +
a_{m}|v_{m}\rangle, \quad \mbox{where} \quad a_{j} \in \mathds{C}.
\end{equation}
The vector $|\phi\rangle$ can be also written in the column matrix
notation
\begin{equation}
|\phi\rangle = \left(\begin{array}{cccc} a_{0} & a_{1} & \cdots &
a_{m}\end{array}\right)^T
\end{equation}
Let $\mathds{C}^{m+1}$ be now regarded as a Hilbert space and
$|\phi\rangle, |\chi\rangle \in \mathds{C}^{m+1}$. The inner
product of the vector $|\phi\rangle$ with the vector
$|\chi\rangle$ will be denoted by $\langle \phi|\chi\rangle$. The
notation $\langle \phi|$ is used for the dual vector to
$|\phi\rangle$. The dual vector $\langle \phi|$ (also called {\it
bra}) is a linear operator $\langle \phi|\colon \mathds{C}^{m+1}
\to \mathds{C}$ defined by $\langle
\phi|(|\chi\rangle)\mathrel{\mathop:}= \langle\phi|\chi\rangle.$
Thus, the inner product requirements imply that
\begin{equation} \label{dual}
\langle\phi| = a^*_{0}\langle v_{1}| + a^*_{1}\langle v_{2}| +
\dots + a^*_{m}\langle v_{m}|.
\end{equation}
The common assumption in quantum computing is to consider Hilbert
space $\mathds{C}^{m+1}$ with an orthonormal basis. Let us denote
the basis as $\{|x\rangle\}_{x=0,1,\dots,m}$ (also called {\it
computational basis}). Let $|\phi\rangle = \sum_{x}b_{x}|x\rangle$
and $|\chi\rangle = \sum_{x}c_{y}|x\rangle$ be the vectors with
respect to the basis $\{|x\rangle\}_{x=0,1,\dots,m}$. Then the
inner product $\langle\phi|\chi\rangle$ can be expressed in terms
of matrix multiplication:
\begin{equation}
\langle\phi|\chi\rangle = \left(\begin{array}{cccc} b^*_{0} &
b^*_{1} & \cdots &
b^*_{m}\end{array}\right)\left(\begin{array}{cccc} c_{0} & c_{1} &
\cdots & c_{m}\end{array}\right)^T.
\end{equation}
In this case, the dual vector $\langle\phi|$ has a row matrix
representation whose entries are complex conjugates of the
corresponding entries of the column matrix representation of
$|\phi\rangle$.

The fundamental concept of quantum information is {\it quantum
bit} ({\it qubit}) described mathematically as a unit vector
$|\varphi\rangle$ in a Hilbert space $\mathds{C}^2$. According to
the notation explained above: \begin{equation} \label{qubit}
|\varphi\rangle = d_{0}|0\rangle + d_{1}|1\rangle,\quad
\mbox{where}\quad d_{0},d_{1} \in \mathds{C} \quad \mbox{and}
\quad |d_{0}|^2 + |d_{1}|^2=1.
\end{equation}
The measurement of a qubit with respect to an orthonormal basis
$\{|w_{0}\rangle, |w_{1}\rangle\}$ (not necessarily in the
computational basis) yields the result $w_{0}$ or $w_{1}$ with
probability $|\langle w_{j}|\varphi\rangle|^2$ leaving the qubit
in the corresponding state $|w_{0}\rangle$ or $|w_{1}\rangle$. In
particular, measuring the qubit given by (\ref{qubit}) with
respect to $\{|0\rangle, |1\rangle\}$ results in the outcome 0
with probability $|d_{0}|^2$ and the outcome 1 with probability
$|d_{1}|^2$, with post-measurement states $|0\rangle$ and
$|1\rangle$, respectively.

Suppose $\mathscr{H}_{1}$ and $\mathscr{H}_{2}$ are Hilbert spaces
with orthonormal bases $\{|x\rangle\}_{x=0,1,\dots,m_{1}}$ and
$\{|y\rangle\}_{y=0,1,\dots,m_{2}}$, respectively. Then the tensor
product $\mathscr{H}_{1} \otimes \mathscr{H}_{2}$ is a~Hilbert
space of $(m_{1}+1)(m_{2}+1)$ dimensionality with the orthonormal
basis $\{|x\rangle \otimes |y\rangle\}$. The matrix representation
of an element $|x\rangle \otimes |y\rangle$ is the Kronecker
product of respective matrix representations of $|x\rangle$ and
$|y\rangle$. In the further part of the paper we use the
abbreviated notation $|x\rangle|y\rangle$ or $|x,y\rangle$ for the
tensor product $|x\rangle \otimes |y\rangle$.

A system of $n$ qubits $|\varphi_{i}\rangle$ is described as a
unit vector $|\psi\rangle$ in the tensor product space
$\bigotimes^n_{j=1}\mathds{C}^2$ that has $2^n$-element
computational basis
\begin{equation}
\{|x_{1}\rangle \otimes |x_{2}\rangle \otimes \dots \otimes
|x_{n}\rangle\}_{x_{j} = 0,1}.
\end{equation}
 Thus, it is described by the vector
\begin{equation}\begin{split}\label{multiple}
&|\psi\rangle =
\sum_{x_{1},x_{2},...,x_{n}}d_{x_{1},x_{2},...,x_{n}}|x_{1},x_{2},\dots,x_{n}\rangle,\\
&\mbox{where} \quad d_{x_{1},x_{2},...,x_{n}} \in \mathds{C} \quad
\mbox{and} \quad
\sum_{x_{1},x_{2},...,x_{n}}|d_{x_{1},x_{2},...,x_{n}}|^2 =
1.\end{split}\end{equation} We say that the state (\ref{multiple})
is {\it separable} if it can be written as $|\psi\rangle =
\bigotimes^n_{i=1} |\varphi_{i}\rangle$ for some
$|\varphi_{i}\rangle \in \mathds{C}^2$, $i=1,2,\dots,n$. The dual
vector $\langle \psi|$ is defined in the same way as in
(\ref{dual}). Similarly, the measurement of the state given by
(\ref{multiple}) with respect to an orthonormal basis
$\{|w_{j}\rangle\}^{2^n}_{j=1}$ yields the result $w_{j}$ with
probability $|\langle w_{j}|\psi\rangle|^2$. Otherwise, the state
$|\psi\rangle$ is called {\it entangled}.\vspace{12pt}

\noindent We use the Dirac notation throughout the whole paper.
However, each of the results below can be easily reconstructed
using the matrix notation.
\section{Normal
representation of extensive games in quantum domain}   From that
moment on, we will consider extensive games with two available
actions at each information set so that we could use only qubits
for convienience. Any game richer in actions can be transferred to
quantum domain by using quantum objects of higher
dimensionality.\vspace{12pt}

\noindent Let us extend the protocol (\ref{fourtuple}) to include
components making it useful for extensive games. Such a quantum
game is specified by a six-tuple:
\begin{equation}
\mathrm{\Gamma}^{\mathrm{QI}} = \left(\mathscr{H}, N,
|\psi_{\mathrm{in}}\rangle, \xi, \{\mathcal{U}_{j}\},
\{E_{i}\}\right) \label{sixtuple}
\end{equation}
where the components are defined as follows:
\begin{itemize}
\item $\mathscr{H}$ is a complex Hilbert space $\bigotimes_{j=1}^m
\mathds{C}^2$ with an orthonormal basis $\mathcal{B}$. \item $N$
is a set of players with the property that $|N| \leq m$. \item
$|\psi_{\mathrm{in}}\rangle$ is the initial state of a quantum
system of $m$ qubits $|\varphi_{1}\rangle,
|\varphi_{2}\rangle,\dots,|\varphi_{m}\rangle$. \item $\xi\colon
\{1,2,\dots,m\} \to N$ is a surjective mapping. A value $\xi(j)$
indicates a~player who carries out a unitary operation on a qubit
$|\varphi_{j}\rangle$. \item For each $j \in \{1,2,\dots,m\}$ the
set $\mathcal{U}_{j}$ is a subset of unitary operators from
$\mathsf{SU}(2)$ that are available for a qubit $j$. A~(pure)
strategy of a player $i$ is a~map $\tau_{i}$ that assigns a
unitary operation $U_{j} \in \mathcal{U}_{j}$ to a qubit
$|\varphi_{j}\rangle$ for every $j \in \xi^{-1}(i)$. The final
state $|\psi_{\mathrm{fin}}\rangle$ when the players have
performed their strategies on corresponding qubits is defined as:
\begin{equation}
|\psi_{\mathrm{fin}}\rangle\mathrel{\mathop:}=
(\tau_{1},\tau_{2},\dots,\tau_{n})|\psi_{\mathrm{in}}\rangle =
\bigotimes_{i\in N}\bigotimes_{j\in \xi^{-1}(i)}U_{j}
|\psi_{\mathrm{in}}\rangle. \label{finalstate}
\end{equation} \item For each $i\in N$ the map $E_{i}$ is a utility (payoff) functional that
specifies a~utility for the player $i$. The functional $E_{i}$ is
defined by the formula:
\begin{equation}\label{eformula}
E_{i} = \sum_{|b\rangle \in \mathcal{B}}v_i(b)|\langle b
|\psi_{\mathrm{fin}}\rangle|^2, ~~ \mbox{where} ~~ v_i(b) \in
\mathds{R}.
\end{equation}
\end{itemize}
There are only two additional components in (\ref{sixtuple}): $N$
and $\xi$, in comparison with (\ref{fourtuple}). They completely
specify qubits to which a player is permitted to apply her unitary
operator. Notice also that the protocol of quantization of
strategic games according to \cite{eisert2} is obtained from
$\left(\mathscr{H}, N, |\psi_{\mathrm{in}}\rangle, \xi,
\{\mathcal{U}_{j}\}, \{E_{i}\}\right)$ by putting $|N| = m = 2$.
We claim that such addition together with appropriate fixed values
$v_{i}(b)$ in (\ref{eformula}) are sufficient for considering an
extensive game in quantum domain (of course, if the assumption
that the tuple $\left(\mathscr{H}, |\psi_{\mathrm{in}}\rangle,
\{\mathcal{U}_{i}\}, \{E_{i}\}\right)$ correctly describes
strategic games in quantum domain is true). The line of thought is
as follows. Any strategic game can be considered as a special case
of an extensive game where players move sequentially but each of
them does not have any knowledge about actions taken by the other
players. In other words, each player in a strategic game has
exactly one information set in which she takes an action. Thus, in
a simple case of $2\times2$ bimatrix game, the scheme
(\ref{fourtuple}), in fact, identifies an operation on a qubit
with player's move made at her unique information set, and then
the individual game outcomes are assigned to appropriate
measurement results. An extensive game can have many information
sets, and more than one of them can be assigned to the same
player. Therefore, our extension of $\left(\mathscr{H},
|\psi_{\mathrm{in}}\rangle, \{\mathcal{U}_{i}\}, \{E_{i}\}\right)$
is aimed at similar identification for extensive games. As
a~result, we obtain that we are able to write in (\ref{sixtuple})
the normal representation of an extensive game.

 Before we formulate the formal statement, notice first
that the tuple (\ref{sixtuple}), in fact, determines some game in
strategic form in the sense of classical game theory. If
$\mathscr{H}$ and $|\psi_{\mathrm{in}}\rangle$ are fixed, each
player $i\in N$ chooses her strategy from a set $\bigotimes_{j\in
\xi^{-1}(i)}\mathcal{U}_{j}$ and then the associated utility
$E_{i}$ is determined. Therefore, it always makes sense to
associate (\ref{sixtuple}) with some $\left( N, \{S_{i}\},
\{u_{i}\} \right)$. Secondly, let us specify a sufficient
condition of equivalence for two strategic games $\left( N,
\{S_{i}\}, \{u_{i}\} \right)$ and $\left( N, \{S'_{i}\},
\{u'_{i}\} \right)$. Namely, if there is a~bijective mapping
$g_{i} \colon S_{i} \to S'_{i}$ for each $i \in N$ such that for
each profile $s \in \prod_{i \in N}S_{i}$ we have $u(s) =
u'(g(s))$ where $g = \prod_{i \in N}g_{i}$, then the games are
\textit{isomorphic} (to find out more about isomorphisms of
strategic games see \cite{harsai}). Now, we can formulate the
following proposition:
\begin{proposition} \label{propositiondupa}
Let $\mathrm{\Gamma} = \left( N, H, P, \{\mathcal{I}_{i}\},
\{u_{i}\} \right)$ be a finite extensive game with two available
actions at each information set. Then there exists a six-tuple
(\ref{sixtuple}) that specifies a game isomorphic to the normal
representation of~$\mathrm{\Gamma}$.
\end{proposition}
\begin{proof}
Let us consider an $n$-player extensive game $\mathrm{\Gamma}$
with $m$ information sets. In addition, let us assume two-element
set of available actions $A(I_{i})$ in each information set
$I_{i}$. We specify components of the tuple
$\mathrm{\Gamma}^{\mathrm{QI}}$ as follows. Let $\mathcal{B}$ be
the computational basis of $\bigotimes_{j=1}^m \mathds{C}^2$ and
let the initial state $|\psi_{\mathrm{in}}\rangle$ be of the form
$|b\rangle$, where $|b\rangle$ is some fixed state of
$\mathcal{B}$. Let us restrict the set of available operators on
$\mathds{C}^2$ to the set of two operators $\{\sigma_{0},
\sigma_{1}\}$ where $\sigma_{0}$ is the identity operator and
$\sigma_{1}$ is the bit-flip Pauli operator. This specification
implies that for any mapping $\xi\colon \{1,2,\dots,m\} \to N$
specified in (\ref{sixtuple}) each strategy profile is an operator
of the form $\bigotimes_{j=1}^m\sigma^{j}$, where $\sigma^{j} \in
\{\sigma_{0}, \sigma_{1}\}$. Thus, for each strategy profile
$\tau$ there is some $|b'\rangle \in \mathcal{B}$ such that
\begin{equation}
 |\psi_{\mathrm{fin}}\rangle =
\bigotimes_{j=1}^m\sigma^{j}|\psi_{\mathrm{in}}\rangle =
|b'\rangle \langle b'| \quad \mbox{and} \quad E_{i}(\tau) =
v_i(b')\;\;\mbox{for}\;\;i \in N. \label{classicfin}
\end{equation}
Let us fix a bijective mapping $\zeta$ between $\{1,2,\dots,m\}$
and the set $\{I_{i}\}$ of all players' information sets of
$\mathrm{\Gamma}$. Since for each player $i$ and history $h \in
I_{i}$ we have $P(h) = i$ we can simply take $P(I_{i}) = i$. Then
the correspondence $\xi \mathrel{\mathop:}= P \circ \zeta$
associates each information set of each player with exactly one
qubit. As $|A(I_{i})| = |\{\sigma_{0}, \sigma_{1}\}| = 2$ and a
number of information sets of $\mathrm{\Gamma}$ is equal to a
number of qubits, a set of strategies $S_{i}$ of the normal
representation of $\mathrm{\Gamma}$ and a strategy set
$\mathcal{T}_{i}$ of quantum game defined by the tuple
(\ref{sixtuple}) are equinumerous (with cardinality equal
$2^{|\xi^{-1}(i)|}$ each) for each $i \in N$. Therefore, for each
$i \in N$, we can define a bijective mapping $g_{i}\colon S_{i}
\to \mathcal{T}_{i}$. These mappings induce the following
bijection between the sets of strategy profiles:
\begin{equation}
g = \left(g_{i}\right)_{i \in N} \colon \prod_{i\in N} S_{i} \to
\prod_{i\in N} \mathcal{T}_{i}.
\end{equation}
The equations in (\ref{classicfin}) imply that for all $i$ we can
select numbers $v_{i}(b) \in u_{i}(Z)$ in (\ref{eformula}) in a
way that $u_{i}(s) = E_{i}(g(s))$ where $u_{i}$ is the utility
function of the normal representation of $\mathrm{\Gamma}$. Such
specification of (\ref{sixtuple}) makes it isomorphic to the
normal representation of $\mathrm{\Gamma}$.
\end{proof}

Many researches on quantum games played via scheme
(\ref{fourtuple}) are based on appropriately fixed basis for a
space, the initial state, and a range of available unitary
operators, in order to obtain interesting properties of a quantum
game.  We will use the two best-known configurations of $\left(
\mathscr{H}, |\psi_{\mathrm{in}}\rangle, \{\mathcal{U}_{i}\},
\{E_{i}\}\right)$: the Marinatto-Weber (MW) scheme
\cite{marinatto} and the Eisert-Wilkens-Lewenstein (EWL) scheme
\cite{eisert} to examine extensive games via the protocol
(\ref{sixtuple}) (see also \cite{fracor1} and \cite{fracor2} for
other applications of these schemes). In the former scheme players
are allowed to use only the identity operator and the bit-flip
Pauli operator. The results superior to classical results are
obtained by manipulating the initial state
$|\psi_{\mathrm{in}}\rangle$. The later scheme allows to use
broader range of unitary operators (including also the whole set
$\mathsf{SU}(2)$). The following examples concern both settings.

To convert the following games into quantum ones,  we use the same
reasoning as in the proof of Proposition \ref{propositiondupa}.
The first example deals with a case where each player operates on
one qubit.
\begin{example} \textup{Let us consider a three player extensive game:
\begin{equation}
\mathrm{\Gamma}_{1} = \left(\{1,2,3\}, H, P, \{I_{i}\}_{i \in
\{1,2,3\}},\{u_i\}_{i \in \{1,2,3\}}\right) \label{seltenhorse}
\end{equation}
determined by the following components:
\begin{itemize}
\item $H = \{\emptyset, (a_0), (a_1), (a_0, c_0), (a_0, c_1),
(a_1, b_0), (a_1, b_1), (a_1, b_0, c_0), (a_1, b_0, c_1)\}$; \item
$P(\emptyset) = 1,~ P(a_1) = 2,~ P(a_0) = P(a_1, b_0) = 3$; \item
$I_{1} = \{\emptyset\},~ I_{2} = \{(a_1)\},~ I_{3} = \{(a_0),
(a_1, b_0)\}$; \item $u_{1,2}(a_0, c_0) = 3$, $u_{1,2}(a_0, c_1) =
u_{1,2}(a_1,b_0,c_1)= 0$, \\ $u_{1,2}(a_1,b_1) = 2$,
$u_{1,2}(a_1,b_0,c_0) = 5$, \\ $u_{3}(a_0, c_0) =
u_{3}(a_1,b_0,c_1)= 1$, $u_{3}(a_0, c_1) = u_{3}(a_1,b_0,c_0) =
0$, $u_{3}(a_1, b_1) =~2$.
\end{itemize}}
\noindent \textup{The game is depicted in Fig.~\ref{figure1}.}
\begin{figure*}
% Use the relevant command to insert your figure file.
% For example, with the graphicx package use
  \includegraphics[width=1\textwidth]{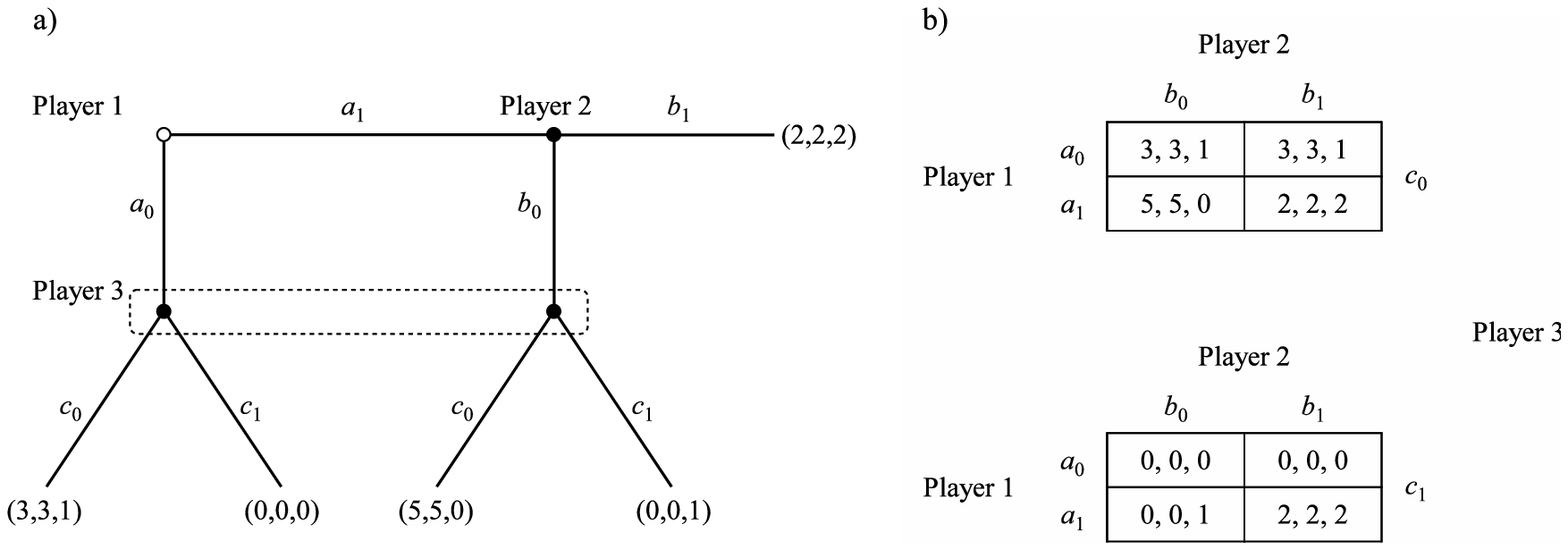}
% figure caption is below the figure
\caption{The modified Selten's Horse game specified by
$\mathrm{\Gamma}_{1}$ a) and its normal representation  b) where
the players choose between rows, column and matrices,
respectively.}
\label{figure1}       % Give a unique label
\end{figure*}
\textup{It is the Selten's Horse game \cite{selten0} with modified
payoffs. Since each of the players has one information set, their
sets of strategies are $\{a_{0}, a_{1}\}$, $\{b_{0}, b_{1}\}$, and
$\{c_{0}, c_{1}\}$, respectively. Profiles: $(a_0, b_1, c_0)$ and
$(a_1, b_1, c_1)$ are the only pure Nash equilibria in this game
and indeed each of them could be equally likely chosen as a
scenario of the game. The utilities for players 1 and 2 assigned
to
 $(a_0,b_1,c_0)$ are higher than the utilities corresponding to $(a_1, b_1, c_1)$ - a desirable profile for player 3. The
uncertainty of a result of the game follows from the peculiar
strategic position of player 3. She could try to affect the
decision of others by announcing before the game starts, that she
is going to take an action $c_{1}$. If the statement of player 3
is credible enough then the history $(a_{1},
b_{1})$ might occur.}\\

\noindent {\it The MW approach.} \, \textup{Let us examine the
Selten's Horse game via the protocol (\ref{sixtuple}). It turns
out that among quantum realizations of the game
$\mathrm{\Gamma}_{1}$ there exist ones that provide the players
with a unique reasonable solution. One of these realizations is
constructed, according to the idea of the MW scheme, as follows:
\begin{equation} \label{gammamw1}
\mathrm{\Gamma}^{\mathrm{MW}}_{1} = \left( \mathscr{H}_{c},
\{1,2,3\}, |\psi_{\mathrm{in}}(\gamma)\rangle,
\mathrm{id}_{\{1,2,3\}}, \{\{\sigma_{0}, \sigma_{1}\}_{i}\},
\{E_{i}\} \right),
\end{equation}
where:
\begin{itemize}
\item $\mathscr{H}_{c}$ is a Hilbert space $\bigotimes_{j=1}^3
\mathds{C}^2$ with the computational basis $\{|x_{1}\rangle
|x_{2}\rangle |x_{3}\rangle\}$, where $x_{j} \in \{0,1\}$ for
$j=1,2,3$; \item the initial state
$|\psi_{\mathrm{in}}(\gamma)\rangle$ takes the form:
\begin{align} |\psi_{\mathrm{in}}(\gamma)\rangle = \cos\frac{\gamma}{2}|000\rangle +
i\sin\frac{\gamma}{2}|111\rangle~~ \mbox{and}~~\gamma \in (0,\pi);
\end{align} \item $\mathrm{id}_{\{1,2,3\}}$ is an
identity mapping defined on $\{1,2,3\}$; \item the payoff
functionals $E_{i}$ are defined as follows:
\begin{equation}\begin{split}
E_{1,2} &=  3\sum_{x_{2}}|\langle
0,x_{2},0|\psi_{\mathrm{fin}}\rangle|^2 + 2\sum_{x_{3}}|\langle
11,x_{3}|\psi_{\mathrm{fin}}\rangle|^2 + 5|\langle 100|\psi_{\mathrm{fin}}\rangle|^2;\\
E_{3} &=  \sum_{x_{2}}|\langle
0,x_{2},0|\psi_{\mathrm{fin}}\rangle|^2 + 2\sum_{x_{3}}|\langle
11,x_{3}|\psi_{\mathrm{fin}}\rangle|^2+ |\langle
101|\psi_{\mathrm{fin}}\rangle|^2.
\end{split}
\end{equation}
\end{itemize}}
\textup{Let us first determine the utilities $E_{i}$ associated
with any profile $\sigma_{\kappa_1} \otimes \sigma_{\kappa_2}
\otimes \sigma_{\kappa_3}$, where $\kappa_{j} \in \{0,1\}$ for $j
= 1,2,3$. The final state $|\psi_{\mathrm{fin}}\rangle$ after the
operation $\sigma_{\kappa_{1}} \otimes \sigma_{\kappa_{2}} \otimes
\sigma_{\kappa_{3}}$ takes the form:
\begin{align}
(\sigma_{\kappa_{1}} \otimes \sigma_{\kappa_{2}} \otimes
\sigma_{\kappa_{3}})|\psi_{\mathrm{in}}\rangle &=
\cos\frac{\gamma}{2}|\kappa_{1},\kappa_{2},\kappa_{3}\rangle
 + i\sin\frac{\gamma}{2}|\overline{\kappa}_{1},\overline{\kappa}_{2},\overline{\kappa}_{3}\rangle,
\end{align}
where $\overline{\kappa}_{j}$ in the negation of $\kappa_{j}$.
Using the last equation and formula (\ref{eformula}) the expected
utilities, for example, for $(\sigma_{1} \otimes \sigma_{0}
\otimes \sigma_{1})$ become:
\begin{align}
E_{1,2}(\sigma_{1} \otimes \sigma_{0} \otimes \sigma_{1}) =
3\sin^2\frac{\gamma}{2}, ~~ E_{3}(\sigma_{1} \otimes \sigma_{0}
\otimes \sigma_{1}) = 1.
\end{align}
All possible values $E_{i}$ are shown in Fig.~\ref{figure2}.
\begin{figure*}
% Use the relevant command to insert your figure file.
% For example, with the graphicx package use
  \includegraphics[width=1\textwidth]{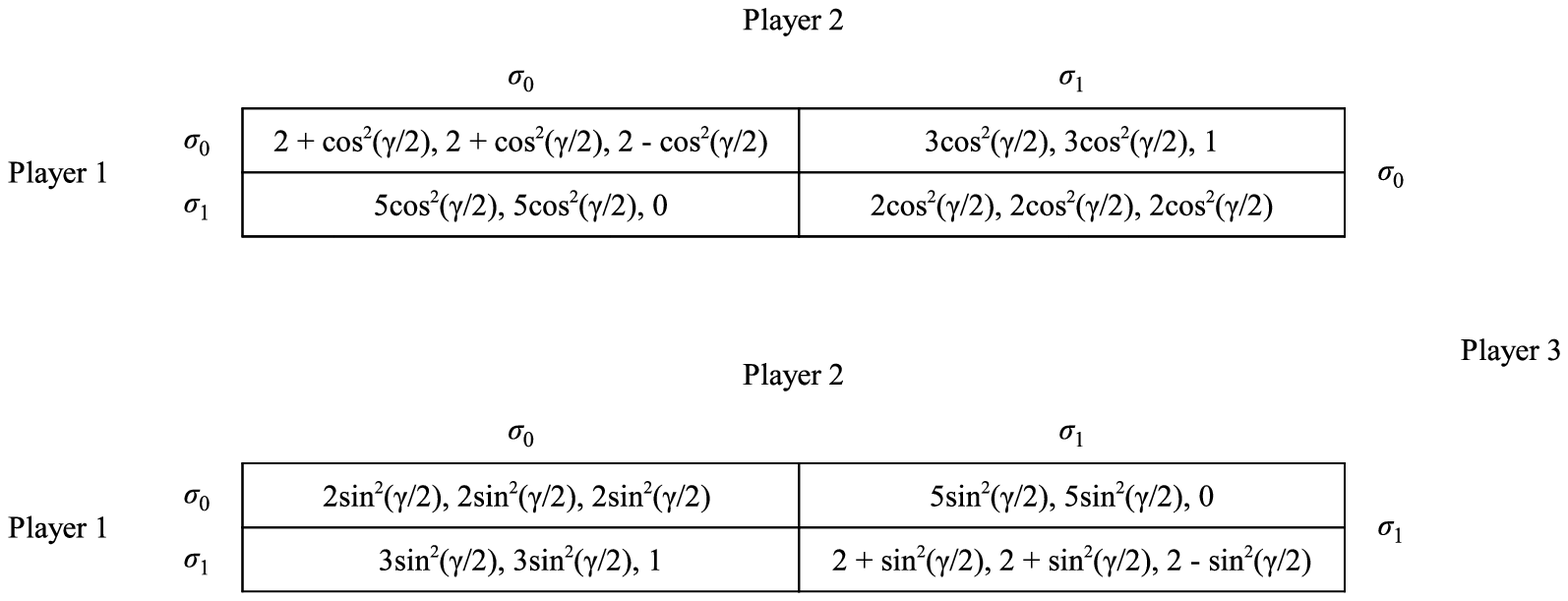}
% figure caption is below the figure
\caption{A game in strategic form induced by the MW approach to
the normal representation of $\mathrm{\Gamma}_{1}$.}
\label{figure2}       % Give a unique label
\end{figure*}
Now, we can analyze the game $\mathrm{\Gamma}^{\mathrm{MW}}_{1}$
like a~classical $2\times 2\times 2$ strategic game. To be fully
precise, such equivalence is assured since we can extend
$\mathrm{\Gamma}^{\mathrm{MW}}_{1}$  to use mixed strategies. A
profile of mixed strategies in the game in Fig.~\ref{figure2}
determines a~probability distribution $\{p_{\kappa_{1},\kappa_{2},
\kappa_{3}}\}$ over profiles $(\sigma_{\kappa_{1}},
\sigma_{\kappa_{2}}, \sigma_{\kappa_{3}})$. Thus, the mixed
strategy outcome in (\ref{gammamw1}) is described simply as an
ensemble $\{p_{\kappa_{1},\kappa_{2}, \kappa_{3}},
(\sigma_{\kappa_{1}} \otimes \sigma_{\kappa_{2}} \otimes
\sigma_{\kappa_{3}})|\psi_{\mathrm{in}}\rangle\}$. Then both the
ensemble and an appropriate profile of mixed strategies in the
game in Fig.~\ref{figure2} generate the same utility outcome.}

\textup{Let us notice now that $\mathrm{\Gamma}^{\mathrm{MW}}_{1}$
is a generalization of the normal representation of
$\mathrm{\Gamma}_{1}$. We are able to get the normal
representation of $\mathrm{\Gamma}_{1}$ out of Fig.~\ref{figure2}
putting the initial state $|\psi_{\mathrm{in}}(0)\rangle$, i.e.,
by putting $\gamma = 0$ into the matrix representation of the game
in Fig.~\ref{figure2}. Moreover, the same is the case for
$|\psi_{\mathrm{in}}(\pi)\rangle$ as well as any initial state
being a basis vector of $\mathscr{H}_{c}$. Then the game
$\mathrm{\Gamma}^{\mathrm{MW}}_{1}$ coincides with the classical
Selten's Horse game up to the order of players' strategies.}

\textup{Now, we examine the six-tuple (\ref{gammamw1}) to find a
reasonable solution for players. Let us determine pure Nash
equilibria in the game $\mathrm{\Gamma}^{\mathrm{MW}}_{1}$ by
solving for each profile $(\sigma_{\kappa_{1}},
\sigma_{\kappa_{2}}, \sigma_{\kappa_{3}})$, where $\kappa_{j} \in
\{0,1\}$ the system of inequalities imposed by the
condition~(\ref{nashequation}). Using values of
$E_{i}(\sigma_{\kappa_{1}}, \sigma_{\kappa_{2}},
\sigma_{\kappa_{3}})$ placed in Fig.~\ref{figure2}, we find that,
for example, the profile $(\sigma_{0}, \sigma_{0}, \sigma_{0})$
 constitutes the
Nash equilibrium if and only if $2\mathrm{cos}^2(\gamma/2)\leq 1$.
Further investigation shows that the profile $(\sigma_{1},
\sigma_{1}, \sigma_{1})$ also fulfills (\ref{nashequation}) with
the requirement $2\mathrm{sin}^2(\gamma/2)\leq 1$ and that there
are no other pure Nash equilibria. Taking into consideration
$\gamma \in (0,\pi)$ we conclude that
\begin{equation} \mathrm{NE_{pure}}(\gamma) = \left\{\begin{array}{lll}
(\sigma_{1}, \sigma_{1}, \sigma_{1}), & \mbox{if} & 0 < \gamma \leq \pi/2;\\
(\sigma_{0}, \sigma_{0}, \sigma_{0}), & \mbox{if} & \pi/2 \leq
\gamma < \pi.
\end{array} \right. \label{roznice}\end{equation}
Let us assume results of the games $\mathrm{\Gamma}_{1}$ and
$\mathrm{\Gamma}^{\mathrm{MW}}_{1}$ to be an equilibrium in pure
strategies. Then formula (\ref{roznice}) shows that each player
can gain from playing game $\mathrm{\Gamma}^{\mathrm{MW}}_{1}$. In
classical case players 1 and 2 can assure themselves 2 utility
units and player 3 can get 1 unit for sure by playing pure
equilibria. All these payoffs are strictly less than the payoffs
corresponding to pure Nash equilibria in
$\mathrm{\Gamma}^{\mathrm{MW}}_{1}$, irrespectively of what is a
value of $\gamma$. Moreover, notice that there is the unique
equilibrium in the game $\mathrm{\Gamma}^{\mathrm{MW}}_{1}$ if
$\gamma \ne \pi/2$ and the same utilities are assigned to both
equilibria in the case $\gamma = \pi/2$. This implies that in the
game $\mathrm{\Gamma}^{\mathrm{MW}}_{1}$ the strategy profile
$(\sigma_{0}, \sigma_{0}, \sigma_{0})$ for $\gamma \in
(\pi/2,\pi)$ and the strategy profile $(\sigma_{1}, \sigma_{1},
\sigma_{1})$ for $\gamma \in (0,\pi/2)$ are reasonable profiles
for all players.}

\textup{An interesting fact worth pointing out is that for
$\gamma$ arbitrary close to 0 or $\pi$ (i.e. for angles defining
the classical game) the equilibrium is unique. This discontinuity
implies possible applications of quantum games to classical game
theory. Namely, the MW approach may serve as a Nash equilibrium
refinement by considering only Nash equilibria that hold out some
slight perturbation of $|\psi_{\mathrm{in}}(0)\rangle$. In fact,
the profile $(\sigma_{1}, \sigma_{1}, \sigma_{1})$ is the unique
pure trembling hand perfect equilibrium in $\mathrm{\Gamma}_{1}$
\cite{selten0} (see also \cite{osborne}, example 252.1). Although
a~further investigation is required, we believe there is a strong
connection between the above method and the Selten's concept of trembling hand equilibrium.}\\

\noindent {\it The~EWL approach.} \, \textup{The second quantum
realization of $\mathrm{\Gamma}_{1}$ is in the spirit of the EWL
protocol. Contrary to the previous one, where the number of
reasonable Nash equilibria was reduced to the unique one, we focus
this time on improving strategic position of only one of the
players. Namely, let us modify the previous quantum game as
follows:
\begin{equation} \label{gammaewl1}
\mathrm{\Gamma}^{\mathrm{EWL}}_{1} = \left(\mathscr{H}_{e},
\{1,2,3\}, |\psi_{\mathrm{in}}\rangle,
\mathrm{id}_{\{1,2,3\}},\{U_{1,2}(\theta,0), U_{3}(\theta,
\alpha)\}, \{E_{i}\}\right),
\end{equation}
\textup{where:}
\begin{itemize}
\item $\mathscr{H}_{e}$ is a Hilbert space $\bigotimes_{j=1}^3
\mathds{C}^2$ with the basis $\{|\psi_{x_{1},x_{2},x_{3}}\rangle
\}_{x_{j \in \{0,1\}}}$ of entangled states defined as follows:
\begin{equation}
|\psi_{x_{1},x_{2},x_{3}}\rangle = \frac{|x_{1},x_{2},x_{3}\rangle
+
i|\overline{x}_{1},\overline{x}_{2},\overline{x}_{3}\rangle}{\sqrt{2}};
\label{basestate}
\end{equation}
\item $|\psi_{\mathrm{in}}\rangle
\mathrel{\mathop:}=|\psi_{000}\rangle$; \item the unitary
strategies are elements of $\{U(\theta, \alpha)\colon \theta \in
[0,\pi], \alpha \in [0,\pi/2]\}$ whose matrix representation with
respect to the computational input and output basis is the
following:
\begin{equation}
U(\theta, \alpha) = \left(\begin{array}{cc}
e^{i\alpha}\cos(\theta/2) &
 i\sin(\theta/2) \\
 i\sin(\theta/2)  & e^{-i\alpha}\cos(\theta/2)
\end{array}\right); \label{twoparameter}
\end{equation}
\item the payoff functional $E_{i}$ is derived from
$\mathrm{\Gamma}^{\mathrm{MW}}_{1}$ with respect to basis states
(\ref{basestate}):
\begin{align} \begin{split}\label{wyplatyewl}E_{1,2} &=
3\sum_{x_{2}}|\langle
\psi_{0,x_{2},0}|\psi_{\mathrm{fin}}\rangle|^2 +
2\sum_{x_{3}}|\langle
\psi_{11,x_{3}}|\psi_{\mathrm{fin}}\rangle|^2
+5|\langle\psi_{100}|\psi_{\mathrm{fin}}\rangle|^2;\\ E_{3} &=
\sum_{x_{2}}|\langle\psi_{0,x_{2},0}|\psi_{\mathrm{fin}}\rangle|^2+
2\sum_{x_{3}}|\langle
\psi_{11,x_{3}}|\psi_{\mathrm{fin}}\rangle|^2 +
|\langle\psi_{101}|\psi_{\mathrm{fin}}\rangle|^2.
\end{split}
\end{align}
\end{itemize}
In this case only the third player is allowed to use unitary
strategies beyond the set of one-parameter operators of the game
$\mathrm{\Gamma}^{\mathrm{EWL}}_{1}$ by using the additional
parameter $\alpha$. We demonstrate now that such extended strategy
set of player 3 significantly improves her strategic position. In
order to see this, let us determine the expected utility $E_{i}$
for each player $i$ that corresponds to a profile of strategies
$\tau = (\theta_{1}, \theta_{2}, (\theta_{3}, \alpha_{3}))$ (since
angles specify strategies of player 1, 2, and 3, we denote them,
for convenience, as $\theta_{1}, \theta_{2}$ and $(\theta_{3},
\alpha_{3})$, respectively). Using formula (\ref{finalstate}) the
final state $|\psi_{\mathrm{fin}}\rangle$ associated with a
profile $\tau$ is the following:
\begin{equation}\label{psifin}
|\psi_{\mathrm{fin}}\rangle = U(\theta_{1},0)\otimes
U(\theta_{2},0)\otimes
U(\theta_{3},\alpha_{3})|\psi_{\mathrm{in}}\rangle =
\frac{1}{\sqrt{2}}\sum_{x\in \{0,1\}^3}\lambda_{x}|x\rangle,
\end{equation}
where
\begin{align}\label{psifindodatek}
\lambda_{x_{1},x_{2},x_{3}} &= i^{\sum
x_{j}}e^{i\overline{x}_{3}\alpha_{3}}\prod_{j}\cos\left(\frac{x_{j}\pi
- \theta_{j}}{2}\right) +(-i)^{\sum
x_{j}}e^{-ix_{3}\alpha_{3}}\prod_{j}\cos\left(\frac{\overline{x}_{j}\pi
- \theta_{j}}{2}\right).
\end{align}
Putting (\ref{psifin}) into formulae (\ref{wyplatyewl}) we obtain
the utility outcomes:
\begin{align}\begin{split}
E_{1,2}(\tau)
&=2\left(\sin^2\frac{\theta_{1}}{2}\sin^2\frac{\theta_{2}}{2}\sin^2\frac{\theta_{3}}{2}
+
\cos^2\frac{\theta_{1}}{2}\cos^2\frac{\theta_{2}}{2}\cos^2\frac{\theta_{3}}{2}\sin^2\alpha_{3}
\right)\\ &\quad+
 \cos^2\frac{\theta_{3}}{2}\cos^2\alpha_{3}\left[3\cos^2\frac{\theta_{1}}{2}
+ \sin^2\frac{\theta_{1}}{2}\left(2 +
3\cos^2\frac{\theta_{2}}{2}\right)\right];\label{quantumgame2utility0}\\
E_{3}(\tau)
&=\cos^2\frac{\theta_{3}}{2}\cos^2\alpha_{3}\left(2\sin^2\frac{\theta_{1}}{2}\sin^2\frac{\theta_{2}}{2}
+ \cos^2\frac{\theta_{1}}{2}\right)\\
&\quad+\cos^2{\frac{\theta_{1}}{2}}\cos^2\frac{\theta_{3}}{2}\sin^2{\alpha_{3}}\left(
1+
\cos^2\frac{\theta_{2}}{2}\right)+\sin^2\frac{\theta_{1}}{2}\sin^2\frac{\theta_{3}}{2}\left(1+\sin^2\frac{\theta_{2}}{2}\right).
\end{split}
\end{align}
As it should be expected in the EWL protocol, we get the classical
game $\mathrm{\Gamma_{1}}$ when the player~3 is also restricted to
use only the strategies $(\theta_{3},0)$. Namely, let us put
$\alpha = 0$, $p\mathrel{\mathop:}=\mathrm{cos}^2(\theta_{1}/2)$,
$q\mathrel{\mathop:}=\mathrm{cos}^2(\theta_{2}/2)$ and
$r\mathrel{\mathop:}=\mathrm{cos}^2(\theta_{3}/2)$ to
Eq.~(\ref{quantumgame2utility0}). Then we obtain
\begin{equation} \begin{split}
E_{1,2}(\theta_{1}, \theta_{2}, (\theta_{3}, 0)) &= 3pr + 5(1-p)qr + 2(1-p)(1-q); \label{classicale1}\\
E_{3}(\theta_{1}, \theta_{2}, (\theta_{3}, 0)) &= pr + (1-p)q(1-r)
+ 2(1-p)(1-q).\end{split}
\end{equation}
Formulae (\ref{classicale1}) are exactly the players expected
payoffs in the classical game $\mathrm{\Gamma}_{1}$ if they choose
their actions $a_{0}, b_{0}$ and $c_{0}$ with probability $p, q,$
and $r,$ respectively. Thus, in the particular case, if $p,q,r \in
\{0,1\}$, we obtain payoffs corresponding to pure strategy
profiles in $\mathrm{\Gamma}_{1}$.}

\textup{Now we solve a problem how player 3 can gain from using
2-parameter operators as her strategies. An interesting feature is
that the game $\mathrm{\Gamma}^{\mathrm{EWL}}_{1}$ keeps the pure
Nash equilibria of the game $\mathrm{\Gamma}_{1}$, i.e., the
profiles: $(0, \pi, (0, 0))$ and $(\pi, \pi, (\pi, 0))$ -
equivalents for the respective Nash equilibria profiles $(a_{0},
b_{1}, c_{0})$ and $(a_{1}, b_{1}, c_{1})$ in
$\mathrm{\Gamma}_{1}$ - are Nash equilibria profiles also in
$\mathrm{\Gamma}^{\mathrm{EWL}}_{1}$. However, unlike in the game
$\mathrm{\Gamma}_{1}$, there is another pure equilibrium $\tau^* =
(0,0,(0, \pi/2))$ where $E_{i}(\tau^*) = 2$ for each player $i$.
This non-equivalence to the classical profile is essential for
strategic position of player 3. She can force the other players to
play strategies from the profile $\tau^*$ (instead of $(0, \pi,
(0, 0))$ - their the most preferred equilibrium) by making an
announcement that she is going to play $\tau^*_{3} = (0,\pi/2)$.
The other players know that this threat is credible enough as the
player 3 does not suffer a loss when she deviates from $(0, 0)$ to
$(0, \pi/2)$, since $E_{3}(0, \pi, (0, 0)) = E_{3}(0, \pi, (0,
\pi/2)) = 1$. However, the opponents of the third player lose 3
utilities, since $E_{1,2}(0, \pi, 0, \pi/2) = 0$. Furthermore,
given $(\theta_{3}, \alpha_{3}) = (0,\pi/2)$ fixed,
 rationality demands that they play strategies dictated by
$\tau^*$ as we have $\mathrm{arg}\max_{\theta_{1},
\theta_{2}}E_{1,2}(\theta_{1},\theta_{2},(0,\pi/2)) = \{(0,0)\}$.
This argumentation allows to treat the profile $\tau^*$ as
reasonable solution of $\mathrm{\Gamma}^{\mathrm{EWL}}_{1}$. Thus,
the strategic position of the player 3 has been significantly
improved in comparison to her classical strategies.}
\label{example1}
\end{example}

The second example is aimed at showing that the proposed scheme of
playing extensive games based on the six-tuple (\ref{sixtuple})
can be applied to extensive games in which some of players have
more than one information set. Unlike in Example~\ref{example1} we
now focus only on converting an extensive game into the form
described by (\ref{sixtuple}) without considering a specific
strategic situation.
\begin{example}
\textup{Let the following extensive game be given:
\begin{equation} \label{drgamma2}
\mathrm{\Gamma}_{2} = \left(\{1,2\}, H, P, \{\mathcal{I}_{i}\},
u\right),
\end{equation}
where
\begin{itemize} \item $H = \{\emptyset,
a_{0}, a_{1}, (a_{0}, b_{0}), (a_{0}, b_{1}), (a_{1}, c_{0}),
(a_{1}, c_{1})\}$; \item $P(\emptyset) = 1$, $P(a_0) = P(a_1) =
2$; \item $\mathcal{I}_{1} = \{\{\emptyset\}\}$, $\mathcal{I}_{2}
= \{\{a_{0}\}, \{a_{1}\}\}$; \:$u(a_{\iota_{1}}, b_{\iota_{2}}) =
O_{\iota_{1},\iota_{2}}$, $\iota_{1},\iota_{2} = 1,2$.
\end{itemize}
Like in the previous example, the game $\mathrm{\Gamma}_{2}$ has
three information sets in which two actions are available.
However, in this case, two information sets represent the
knowledge of player~2. Thus, she specifies an action at each of
them. The game is illustrated in Fig.~\ref{figure3}.
\begin{figure*}
% Use the relevant command to insert your figure file.
% For example, with the graphicx package use
  \includegraphics[width=1\textwidth]{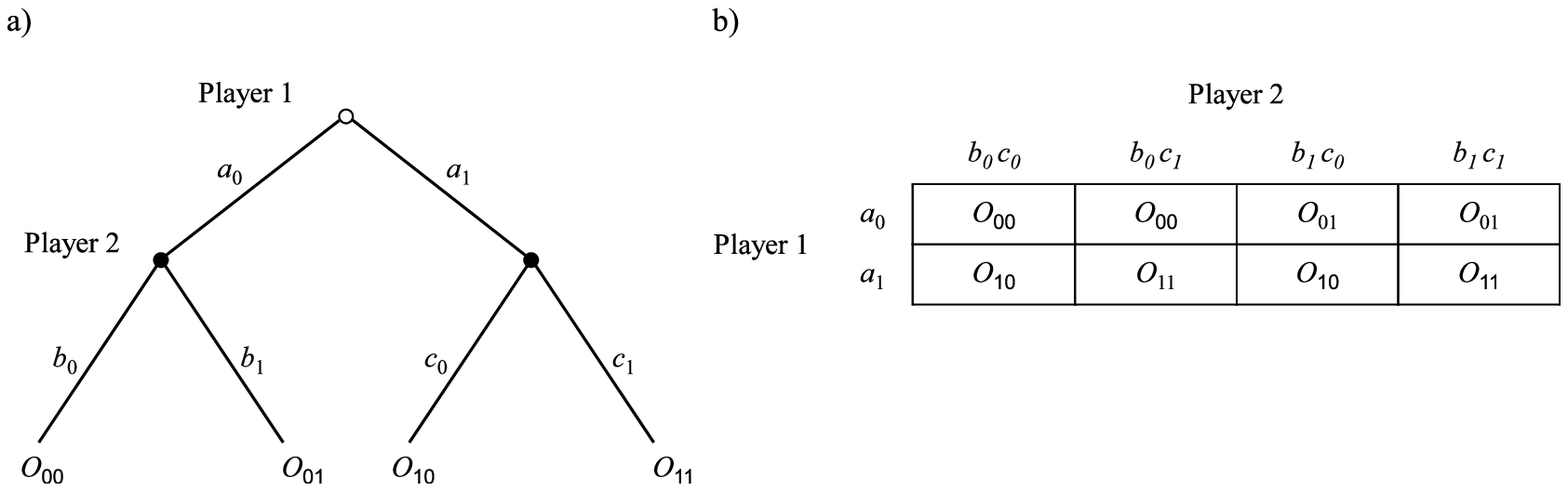}
% figure caption is below the figure
\caption{The game $\mathrm{\Gamma}_{2}$: an extensive form a) and
a normal form b).}
\label{figure3}       % Give a unique label
\end{figure*}
Since there will not be a need to use individual payoffs for
players, we assign an outcome $O_{\iota_{1},\iota_{2}}$ to each of
terminal histories in $H$ for convenience.}\\

\noindent {\it The EWL approach.} \, \textup{Let us first put the
game (\ref{drgamma2}) into the form described by the six-tuple
(\ref{sixtuple}) via the EWL approach. Using the same line of
reasoning as in the proof of Proposition~\ref{propositiondupa},
the number of qubit on which a player is allowed to operate has to
agree with the number of her information sets. To define the
outcome functional $E$ for the quantum game we associate
particular outcomes $O_{\iota_{1},\iota_{2}}$ with appropriate
basis states of $\{|\psi_{x_{1},x_{2},x_{3}}\rangle\}$. For
example, we identify the outcome $O_{10}$ of the game
$\mathrm{\Gamma}_{2}$ with the basis states
$\{|\psi_{1,x_{2},0}\rangle\}_{x_{2} = 0,1}$ measured on
$|\psi_{\mathrm{fin}}\rangle$, and the outcome $O_{00}$ with the
basis states $\{|\psi_{00,x_{3}}\rangle\}_{x_{3} = 0,1}$.
Formally, the EWL approach to (\ref{drgamma2}) is a six-tuple
\begin{equation} \label{ewlapproach2}
\mathrm{\Gamma}^{\mathrm{EWL}}_{2} = \left(\mathscr{H}_{e},
\{1,2\}, |\psi_{\mathrm{in}}(\pi/2)\rangle,
\xi,\{\mathcal{U}_{i}\}, E\right)
\end{equation}
defined by the following components:
\begin{itemize} \item
the map $\xi$ on $\{1,2,3\}$ given by the formula: $\xi(j) =
\left\{\begin{array}{lll}
1, & \mbox{if} & j=1;\\
2, & \mbox{if} & j \in \{2,3\}.
\end{array}; \right.$ \item the set $\mathcal{U}_{i}$ of unitary operators such that $\{U(\theta,0)\} \subseteq
\mathcal{U}_{i} \subseteq \mathsf{SU}(2)$ for $i = 1,2$; \item $E$
is the outcome functional of the form:
\begin{align}
E = \sum_{\iota_{2},x_{3}=0,1}O_{0,\iota_{2}}|\langle
\psi_{0,\iota_{2},x_{3}}|\psi_{\mathrm{fin}}\rangle|^2 +
\sum_{\iota_{2},x_{2}=0,1}O_{1,\iota_{2}}|\langle
\psi_{1,x_{2},\iota_{2}}|\psi_{\mathrm{fin}}\rangle|^2.
\end{align}
\end{itemize}
Let us prove that (\ref{ewlapproach2}) generalizes
(\ref{drgamma2}). Following the definition of
$\mathrm{\Gamma}^{\mathrm{EWL}}_{2}$, the strategy set of player 1
is simply $\mathcal{U}_{1}$, and player 2 chooses her strategies
from the set $\mathcal{U}_{2} \otimes \mathcal{U}_{2}$ since she
operates on the second and the third qubit. Therefore, the final
state $|\psi_{\mathrm{fin}}\rangle$ in the game
$\mathrm{\Gamma}^{\mathrm{EWL}}_{2}$ takes the form of
$\bigotimes^3_{j=1}U_{j}|\psi_{\mathrm{in}}\rangle$, where $U_{j}
\in \mathcal{U}_{\xi(j)}$. Let us assume that the players apply
unitary operators form the set $\{U(\theta,0)\}$. Then the final
state is represented by Eq.~(\ref{psifin}) and
Eq.~(\ref{psifindodatek}) for $\alpha_{3} =0$. It implies that the
expected outcome $E(\theta_{1},(\theta_{2},\theta_{3}))$ equals
\begin{align}\label{ostatnieequation}
E(\theta_{1},(\theta_{2},\theta_{3})) &=
\left(O_{00}\cos^2\frac{\theta_{2}}{2}+
O_{01}\sin^2\frac{\theta_{2}}{2}\right)\cos^2\frac{\theta_{1}}{2}\notag\\
&\quad+ \left(O_{10}\cos^2\frac{\theta_{3}}{2}+
O_{11}\sin^2\frac{\theta_{3}}{2}\right)\sin^2\frac{\theta_{1}}{2}.
\end{align}
By substitution $p\mathrel{\mathop:}= \cos^2\theta_{1}/2$,
$q\mathrel{\mathop:}= \cos^2\theta_{1}/2$, and
$r\mathrel{\mathop:}= \cos^2\theta_{1}/2$,
Eq.~(\ref{ostatnieequation}) shows the expected outcome in game
$\mathrm{\Gamma}_{2}$ when player 1 chooses $a_{0}$ with
probability $p$ and player~2 chooses $b_{0}$ and $c_{0}$ with
probability $q$ and $r$, respectively. To sum up, the six-tuple
(\ref{ewlapproach2}) indeed allows to describe the quantum
extension of the game
$\mathrm{\Gamma}_{2}$ within the EWL approach.}\\

\noindent {\it The MW approach.} \, \textup{In a similar way, we
rewrite the game $\mathrm{\Gamma}_{2}$ using the MW approach by
replacing the components $\mathscr{H}_{e}$,
$\{|\psi_{x_{1},x_{2},x_{3}}\rangle\}$, and $\{\mathcal{U}_{i}\}$
with $\mathscr{H}_{c}$, $\{|x_{1},x_{2},x_{3}\rangle\}$, and
$\{\{U(0,0),U(\pi,0)\}_{i}\}$, respectively. Then similar analysis
to that in Example~\ref{example1} shows that the game given by
$\mathrm{\Gamma}^{\mathrm{MW}}_{2}$ coincides with
$\mathrm{\Gamma}_{2}$ if $|\psi_{\mathrm{fin}}\rangle =
|\psi_{\mathrm{fin}}(0)\rangle$.}
\end{example}
\section{Conclusion}
We have shown that the six-tuple $\left( \mathcal{H}, N,
\rho_{\mathrm{in}}, \xi, \{\mathcal{U}_{j}\}, \{E_{i}\} \right)$
allows to study extensive games using quantum information
language. Although proposed scheme is suitable only for a normal
representation in which some features of corresponding classical
game in extensive form are lost, it yields on valuable information
about how passing to quantum domain influences a course of
extensive games.  The examples we studied have shown that an
extensive game played with the use of both the MW approach and the
EWL approach substantially differs from this game played
classically. Furthermore, the quantum schemes may yield to the
players' significant advantages in the form of better strategic
positions and pointing out reasonable solutions, as it often
happens in the area of strategic games.

% BibTeX users please use one of
%\bibliographystyle{spbasic}      % basic style, author-year citations
%\bibliographystyle{spmpsci}      % mathematics and physical sciences
%\bibliographystyle{spphys}       % APS-like style for physics
%\bibliography{}   % name your BibTeX data base

% Non-BibTeX users please use

\end{document}